\documentclass[sigconf]{acmart}

\usepackage{algorithm}
\usepackage[noend]{algpseudocode}
\usepackage{amsmath}
\usepackage{cleveref}
\usepackage{graphicx}
\usepackage{xcolor}
\usepackage{amsfonts}
\usepackage{amsthm}
\usepackage{pifont}
\usepackage{xspace}
\usepackage{mathtools}
\usepackage{enumitem}
\usepackage{lineno}
\usepackage{thmtools}
\usepackage{thm-restate}
\usepackage{hhline}% http://ctan.org/pkg/hhline

\usepackage{thmtools}
\usepackage{thm-restate}

\usepackage{hyperref}

\setlist[description]{leftmargin=\parindent,labelindent=\parindent}

\newcommand{\alglinenoNew}[1]{\newcounter{ALG@line@#1}}
\newcommand{\alglinenoPop}[1]{\setcounter{ALG@line}{\value{ALG@line@#1}}}
\newcommand{\alglinenoPush}[1]{\setcounter{ALG@line@#1}{\value{ALG@line}}}

\newcommand{\sys}{DAG-Rider\xspace}

\newcommand{\com}[1]{}

\algdef{SE}[Receiving]{Receiving}{EndReceiving}[1]{\textbf{upon
		receiving}\ #1\ \algorithmicdo}{\algorithmicend\ \textbf{}}%
\algtext*{EndReceiving}

\algdef{SE}[Upon]{Upon}{EndUpon}[1]{\textbf{upon}\ #1\ \algorithmicdo}{\algorithmicend\ \textbf{}}%
\algtext*{EndUpon}

\newcommand\StateX{\Statex\hspace{\algorithmicindent}}
\newcommand\StateXX{\StateX\hspace{\algorithmicindent}}
\algrenewcommand\textproc{}% Used to be \textsc

\newtheorem{lemma}{Lemma}
\newtheorem{claim}{Claim}
\theoremstyle{definition}
\newtheorem{definition}{Definition}[section]

\crefname{table}{table}{tables}
\Crefname{table}{Table}{Tables}
\crefname{algocf}{alg.}{algs.}
\Crefname{algocf}{Alg.}{Algs.}
\Crefname{figure}{Fig.}{Figs.}
\crefname{figure}{fig.}{figs.}
\crefname{claim}{claim}{claims}
\Crefname{claim}{Claim}{Claims}
\crefformat{chapter}{\S#2#1#3}
\crefmultiformat{chapter}{\S\S#2#1#3}{and~#2#1#3}{, #2#1#3}{, and~#2#1#3}

\crefformat{section}{\S#2#1#3}
\crefmultiformat{section}{\S\S#2#1#3}{and~#2#1#3}{, #2#1#3}{, and~#2#1#3}

\makeatletter
\NewEnviron{restatethis}[1]{%
	\protected@write\@auxout{}{%
		\string\@restatetheorem{#1}{\detokenize\expandafter{\BODY}}%
	}%
	\begin{lemma}\label{#1}\BODY\end{lemma}%
}
\newcommand{\@restatetheorem}[2]{%
	\expandafter\gdef\csname restatethis@#1\endcsname{#2}%
}
\newcommand{\restate}[1]{%
	\begingroup
	\renewcommand{\thetheorem}{\ref{#1}}%
	\begin{lemma}\csname restatethis@#1\endcsname\end{lemma}%
	\endgroup
}
\makeatother

\usepackage{lipsum} % Lorem Ipsum for sample file, do not use for generating your manuscript

\copyrightyear{2021}
\acmYear{2021}
\setcopyright{acmcopyright}
\acmConference[PODC '21] {Proceedings of the 2021 ACM Symposium on Principles of Distributed Computing}{July 26--30, 2021}{Virtual Event, Italy.}
\acmBooktitle{Proceedings of the 2021 ACM Symposium on Principles of Distributed Computing (PODC '21), July 26--30, 2021, Virtual Event, Italy}
\acmPrice{15.00}
\acmISBN{978-1-4503-8548-0/21/07}
\acmDOI{10.1145/3465084.3467905}
\begin{document}
%\fancyhead{}
\settopmatter{printfolios=true}
\title{All You Need is DAG}

\author{Idit Keidar}
%\email{first.author@email.com}
\affiliation{%
  \institution{Technion}
%  \city{City}         % ACM template requires city and
  \country{}   % country for affiliation
}

\author{Eleftherios Kokoris-Kogias}
%\email{first.author@email.com}
\affiliation{%
	\institution{IST Austria and Novi Research}
	%  \city{City}         % ACM template requires city and
	  \country{}   % country for affiliation
}

\author{Oded Naor}
\authornote{Oded Naor is grateful to the Technion Hiroshi
	Fujiwara Cyber-Security Research Center for providing a research grant. Part of Oded's work was done while
	at Novi Research.}
%\email{first.author@email.com}
\affiliation{%
	\institution{Technion}
	%  \city{City}         % ACM template requires city and
	  \country{}   % country for affiliation
}

\author{Alexander Spiegelman}
%\email{first.author@email.com}
\affiliation{%
	\institution{Novi Research}
	%  \city{City}         % ACM template requires city and
	  \country{}   % country for affiliation
}

\begin{abstract}
	We present \emph{\sys}, the first asynchronous Byzantine Atomic
	Broadcast protocol that achieves optimal resilience, optimal amortized
	communication complexity, and optimal time complexity.
	\sys is post-quantum safe and ensures that all values proposed by
	correct processes eventually get delivered.
	We construct \sys in two layers: In the first layer, processes
	reliably broadcast their proposals and build a structured Directed
	Acyclic Graph (DAG) of the communication among them.
	In the second layer, processes locally observe their DAGs and totally
	order all proposals with no extra communication.
\end{abstract}

% keywords, ACM classification and conference information can be omitted for submission

\maketitle

% ----------------------------------------------------------------

% title commands are \section, \subsection, \subsubsection and \paragraph

\sloppy
\section{Introduction}

The amplified need in scalable geo-replicated Byzantine fault-tolerant
reliability systems has motivated an enormous amount of study on the
Byzantine State Machine Replication (SMR)
problem~\cite{castro1999practical, kotla2007zyzzyva}.
Many variants of the problem were defined in recent
years~\cite{yin2019hotstuff,lev2020fairledger,kelkar2020order} to
capture the needs of blockchain systems.
To address the fairness issues that naturally arise in
interorganizational deployments,  
we focus on the classic long-lived Byzantine Atomic Broadcast (BAB)
problem~\cite{correia2006consensus, cachin2001secure}, which in
addition to total order and progress also guarantees
that \emph{all} proposals by correct processes are eventually
included.

Up until recently, asynchronous  protocols for the Byzantine consensus
problem~\cite{kapron2010fast,canetti1993fast,cachin2001secure} have been considered too
costly or complicated to be used in practical SMR solutions. 
However, two recent single-shot Byzantine consensus
papers, VABA~\cite{abraham2019vaba} and later Dumbo~\cite{lu2020dumbo},
presented asynchronous solutions with (1) optimal resilience, (2)
expected constant time complexity, and (3) optimal quadratic
communication and optimal amortized linear communication complexity (for the latter).
In this paper, we follow this recent line of work and present
\emph{\sys}: the first asynchronous BAB protocol with optimal resilience, optimal round
complexity, and optimal amortized communication complexity.
In addition, given a perfect shared coin abstraction, our protocol does
not use signatures and does not rely on asymmetric cryptographic
assumptions.
Therefore, when using a deterministic threshold-based coin
implementation with an information theoretical agreement
guarantee~\cite{cachin2005Constantinople,
loss2018combining}, the safety properties of our BAB
protocol are post-quantum secure.

\paragraph{Overview}

We construct \emph{\sys} in two layers: a communication layer and
a zero-overhead ordering layer.
In the communication layer, processes reliably broadcast their proposals
with some meta-data that help them form a \emph{Directed Acyclic
Graph (DAG)} of the messages they deliver.
That is, the DAG consists of rounds s.t.\ every process broadcasts at
most one message in every round and every message has $O(n)$ references
to messages in previous rounds, where $n$ is the total number of
processes.
The ordering layer does not require any extra communication.
Instead, processes observe their local DAGs and
with the help of a little randomization (one coin flip per $O(n)$
decisions on values proposed by different processes) locally order all
the delivered messages in their local DAGs.

A nice feature of \sys is that the propose operation
is simply a single reliable broadcast.
The agreement property of the reliable broadcast ensures that all
correct processes eventually see the same DAG.   
Moreover, the validity property of the reliable broadcast guarantees
that all broadcast messages by correct processes are eventually
included in the DAG.
As a result, in contrast to the VABA and Dumbo protocols that
retroactively ignore half the protocol messages and commit one value
out of $O(n)$ proposals, \sys does not waste any of the messages
and all proposed values by correct processes are eventually ordered (i.e., there is no
need to re-propose).
  
\paragraph{Complexity}

We measure time complexity as the asynchronous 
time~\cite{canetti1993fast} required to commit
$O(n)$ values proposed by different correct processes, and we measure
communication complexity by the number of bits processes send to commit
a single value.
To compare \sys to the state-of-the-art asynchronous
Byzantine agreement protocols, we consider SMR implementations that run an
unbounded sequence of the VABA or Dumbo protocols to independently
agree on every slot.
To compare apples to apples in respect to our time complexity
definition, we allow VABA and Dumbo based SMRs to run up to $n$ slots
concurrently.
Note, however, that for execution processes must output the slot
decisions in a sequential order (no gaps).
Therefore, based on the proof in~\cite{ben2003resilient}, the time
complexity of VABA and Dumbo based SMRs is $O(\log(n))$.
\Cref{table:intro} compares \sys to VABA and Dumbo based SMRs.

\begin{table*}[]
\centering
\setlength\doublerulesep{0.5pt}
\begin{tabular}{ |c|c|c|c|c| }
    \hline
        & Communication  &Expected time  & Post-Quantum & Eventual \\
        & Complexity & Complexity  & Safety &
        Fairness \\
    \hhline{=====} 
        VABA SMR & $O(n^2)$ & $O(\log(n))$ & no & no  \\ 
    \hline
        Dumbo SMR & $\text{amortized } O(n)$ & $O(\log(n))$ & no & no 
        \\
    \hhline{=====}
        \sys  + \cite{bracha1987asynchronous} &\text{amortized } $O(n^2)$ & $O(1)$ & yes & yes  \\
            \hline
        \sys  + \cite{guerraoui2019scalable} &\text{amortized} $O(n \log (n))$ &
        $O \left(\frac{\log(n)}{\log (\log (n))} \right)$ & yes  &
        (1-$\epsilon$)-fair
        \\
            \hline
        \sys  + \cite{cachin2005asynchronous} &\text{amortized } O(n) & $O(1)$ & yes & yes  \\ 
    \hline 
               
\end{tabular}
\caption{A comparison between our protocol with different reliable
broadcast instantiations and VABA and Dumbo based SMR protocols. }
\label{table:intro}
\end{table*}

Since our protocol uses a reliable broadcast abstraction as a basic
building block, different instantiations yield different
complexity. 
For example, if we use the classic Bracha broadcast~\cite{bracha1987asynchronous} to propose
a single value in each message, we get a communication complexity of
$O(n^3)$ per decision.
This is because the Bracha broadcast complexity is $O(n^2)$, and in
order to form a DAG each message has to include an $O(n)$ references to
previous messages.
If we are willing to allow a probability $\epsilon$ to violate progress,
then we can use Guerraoui et al.'s broadcast
protocol~\cite{guerraoui2019scalable} and reduce the complexity to
$O(n^2 \log(n))$ per decision.

Now, just as Dumbo amortizes VABA's communication complexity
from quadratic to linear by using batching and adding a phase of erasure
coding to more economically distribute the data, we can amortize our
communication complexity to be linear per decision as well.
First, since we are anyway including a vector of $O(n)$ references in
every broadcast, batching $O(n)$ proposals in each message shaves a factor of $n$
of the total communication complexity even with Bracha broadcast.
To arrive at the optimal linear complexity, we can replace the reliable
broadcast with the asynchronous verifiable information dispersal
of Cachin and Tessaro~\cite{cachin2005asynchronous}.
The communication complexity of that protocol is $O(n^2\log(n) + n|V|)$,
where $|V|$ is the message size, which allows us to batch $O(n \log (n))$
proposals to achieve optimal amortized communication complexity.

A final feature of our protocol, which is sometimes
underestimated and cannot be presented in a table, is elegance: (1) \sys's modularity clearly separates the communication layer from
the ordering logic; (2) the reliable broadcast abstraction's different
instantiations yield protocols with different trade-offs, and; (3) the
entire detailed pseudocode of the ordering logic spans less than 30
lines.

The rest of this paper is structured as follows: \Cref{sec:model} describes the model and the building blocks used for \sys; \Cref{sec:problemDef} formally defines the BAB problem; \Cref{sec:DAGconstruction} describes the DAG construction layer; \Cref{sec:VAB} specifies the \sys protocol on top of the DAG layer; \Cref{sec:analysis} proves the correctness of the protocol and analyzes its performance; \Cref{sec:related} describes related work; and lastly, \Cref{sec:conclusion} concludes the paper.

\section{Model and Building Blocks}
\label{sec:model}

The system consists of a set $\Pi = \left\{ p_1, \ldots, p_n \right\}$
of $n$ processes,  up to $f < n/3$ of which can act arbitrarily, i.e., be \emph{Byzantine}.
For simplicity, we consider a total of $n = 3f+1$ processes.
The link between every two correct processes is reliable.
Namely, when a correct process sends a message to another correct process, the message eventually arrives and the recipient can verify the sender's identity.
The communication is asynchronous, i.e., there is no bound on the message delivery time.
We consider an adaptive adversary that can dynamically corrupt up to $f$
processes during the run.
Once the adversary corrupts a process, it can drop undelivered messages previously sent from
that process to others.
The adversary controls the arrival times of messages.
As part of the construction, we use two building blocks: a reliable broadcast layer and a delayed global perfect coin, which we describe next.

\paragraph{Reliable broadcast}
There are known algorithms such as Bracha
broadcast~\cite{bracha1987asynchronous} to realize the reliable
broadcast abstraction in the asynchronous network model.
There are also efficient gossip
protocols~\cite{karp2000randomized, boyd2006randomized,
bortnikov2009brahms, guerraoui2019scalable} that provide reliable
broadcast whp at a sub-quadratic communication cost in the number of
processes, and asynchronous verifiable information dispersal
protocols~\cite{lu2020dumbo,cachin2005asynchronous} that use erasure
codes to efficiently batch the broadcast values.

Since we are interested in constructing an asynchronous Atomic
Broadcast that satisfies liveness with probability $1$, we define the
reliable rebroadcast abstraction accordingly to allow the use of
efficient gossip protocols.
Formally, each sender process $p_k$ can send messages by calling
$\textit{r\_bcast}_k(m,r)$, where $m$ is a message, $r \in \mathbb{N}$
is a round number.
Every process $p_i$ has an output $\textit{r\_deliver}_i(m,r,p_k)$,
where $m$ is a message, $r$ is a round number, and $p_k$ is the process
that called the corresponding $\textit{r\_bcast}_k(m,r)$.
The reliable broadcast abstraction guarantees the following properties:
\begin{description}
    \item[Agreement] If a correct processes $p_i$ outputs
    $\textit{r\_deliver}_i(m,r,p_k)$, then every other correct process
    $p_j$ eventually outputs $\textit{r\_deliver}_j(m,r,p_k)$ with probability
    $1$.
    \item[Integrity] For each round $r \in \mathbb{N}$ and process $p_k
    \in \Pi$, a correct process $p_i$ outputs
    $\textit{r\_deliver}_i(m,r,p_k)$ at most once regardless of~$m$.
%     \item[Integrity] If a correct process $p_i$ outputs
%     $\textit{r\_deliver}_i(m,r,p_k)$ and
%     $\textit{r\_deliver}_i(m',r',p_l)$, then $(r,p_k) \neq (r',p_l)$. 
    \item[Validity] If a correct process $p_k$ calls
    $\textit{r\_bcast}_k(m,r)$, then every correct processes $p_i$
    eventually outputs $\textit{r\_deliver}_i(m,r,k)$ with probability
    $1$.
\end{description}

\paragraph{Global perfect coin}
We use a \emph{global perfect coin}, which is unpredictable by the
adversary. 
An instance $w$, $w \in \mathbb{N}$, of the coin is invoked by process $p_i \in \Pi$ by calling $\textit{choose\_leader}_i(w)$.
This call returns a process $p_j \in \Pi$, which is the chosen leader for instance $w$.
Let $X_w$ be the random variable that represents the probability that the coin returns process $p_j$ as the return value of the call $\textit{choose\_leader}_i(w)$.
The global perfect coin has the following guarantees:
\begin{description}
    \item[Agreement] If two correct processes call $\textit{choose\_leader}_i(w)$ and $\textit{choose\_leader}_j(w)$ with respective return values $p_1$ and $p_2$, then $p_1=p_2$.
    \item [Termination] If at least $f+1$ processes call $\textit{choose\_leader}(w)$, then every $\textit{choose\_leader}(w)$ call eventually returns.
    \item[Unpredictability] As long as less than $f+1$ processes call $\textit{choose\_leader}(w)$, the return value is indistinguishable from a random value except with negligible probability $\epsilon$.
    Namely, the probability $pr$ that the adversary can guess the returned
    process $p_j$ of the call $\textit{choose\_leader}(w)$ is $pr \leq
    \Pr [X_w = p_j] + \epsilon$. 
    \item[Fairness] The coin is fair, i.e., $\forall w \in \mathbb{N}, \forall p_j \in \Pi \colon \Pr[X_w = p_j] = 1/n$.
\end{description}

Such coins were used as part of previous Byzantine Agreement
protocols such as~\cite{abraham2019vaba,
cachin2005Constantinople, blum2020byzantine, lu2020dumbo}.
Implementation examples can be found in~\cite{cachin2005Constantinople,
loss2018combining}.
One way to implement a global perfect coin is by using PKI and a
threshold signature scheme~\cite{libert2016born, boneh2001short,
shoup2000practical} with a threshold of $(f+1)$-of-$n$.
When a process invokes an instance $w$ of the coin, it signs $w$ with
its private key and sends the share to all the processes.
Once a process receives $f+1$ shares, it can combine them to get the
threshold signature and hash it to get a random process.
Since the threshold signature value is deterministically determined
by the instance name $w$ such that any $f+1$ shares reveal it
(e.g., the schema in~\cite{shoup2000practical} is based on Shamir's
secret sharing~\cite{shamir1979share}), the coin is perfect (all process agree on the
leader) and its agreement property has information theoretical guarantee.
However, to ensure unpredictability, the PKI must be trusted to ensure that
the adversary cannot generate enough shares to reveal the randomness
before a correct process produces them.
Usually, one assumes that a trusted dealer is used to set up the random
keys for all processes.
However, this assumption can be relaxed by executing an $O(n^4)$
message complexity Asynchronous Distributed Key Generation
protocol~\cite{kokoris2020asynchronous}.
Either way, this scheme remains unpredictable only if the adversary is
computationally bounded.
However, since \sys relies on the unpredictability property of the coin
only for liveness, its safety properties are post-quantum secure.

\section{Problem Definition}
\label{sec:problemDef}
The problem we solve is \emph{Byzantine Atomic Broadcast (BAB)}, which
allows processes to agree on a sequence of messages as needed for State
Machine Replication (SMR).
% To capture the practical settings of a system, we require
% \emph{external validity}~\cite{cachin2001secure}, whereby a well-known
% predicate ${\textit{validate}(tx) \in \{ \textsf{True}, \textsf{False}
% \}}$ returns true if and only if a transaction $tx$ is eligible to be
% included in the SMR.
% For instance, in a blockchain system that is used to log payments
% between users, a transaction is valid if there is no double-spending and the
% signature is correct.
% We say that a value $v$ is \emph{valid} if $\textit{validate}(v)$
% returns true.
%
Due to the FLP result~\cite{fischer1985impossibility}, BAB cannot be
solved deterministically in the asynchronous setting, and therefore we
use the global perfect coin to provide randomness that ensures liveness
with probability 1.
%
% Atomic broadcast is usually defined in terms of broadcast and deliver events.
To avoid confusion with the events of the underlying reliable broadcast
abstraction, we name the broadcast and deliver events of BAB as
$\textit{a\_bcast}(m,r)$ and $\textit{a\_deliver}(m,r,k)$,
respectively, where $m$ is a message, $r \in \mathbb{N}$
is a sequence number, and $p_k \in \Pi$ is a process. 
The purpose of the sequence numbers is to distinguish between messages
broadcast by the same process.
For simplicity of presentation, we assume that each process broadcasts
infinitely many messages with consecutive sequence numbers.

\begin{definition} [Byzantine Atomic Broadcast] Each correct
process $p_i \in \Pi$ can call $\textit{a\_bcast}_i(m,r)$ and output
$\textit{a\_deliver}_i(m,r,k)$, $p_k \in \Pi$.
% We say that $p_i$ \emph{proposes} a value $v$ when 
% $\textit{propose}_i(v)$ is called, and \emph{decides $v$ in slot $\ell$} when
% $\textit{decide}_i(\ell, v)$ is output.
A Byzantine Atomic Broadcast protocol satisfies reliable broadcast
(agreement, integrity, and validity) as well as:
\begin{description}
    \item[Total order] If a correct process $p_i$ outputs
    $a\_deliver_i(m,r,k)$ before $a\_deliver_i(m',r',k')$, then no correct 
    process $p_j$ outputs $a\_deliver_j(m',r',k')$ without first outputting
    $a\_deliver_j(m,r,k)$.
\end{description}
\end{definition}

In the context of Byzantine SMR (e.g., blockchains), the BAB
abstraction support the separation between sequencing of transactions
and execution as done in~\cite{androulaki2018hyperledger}.
BAB provides a mechanism to propose transactions and totally order
them, and an execution engine will have to validate the
transactions before applying them to the SMR.

Moreover, note that our BAB definition provides a
stronger guarantee than the one provided by the sequencing protocols
realized in most Byzantine SMR systems.
Our validity property requires that all messages broadcast by correct
processes are eventually ordered (with probability $1$), whereas most
Byzantine SMR protocols
(i.g.,~\cite{nakamoto2008bitcoin,yin2019hotstuff,castro1999practical}
require that in an infinite run, an infinite number of decisions are
made, but some proposals by correct processes can be ignored.
In addition, it is important to note that our BAB protocol satisfies
chain quality. 
That is, for every prefix of ordered messages of size $(2f+1)r$, $r \in
\mathbb{N}$, at least $(f+1)r$ were broadcast by correct processes.

\paragraph{Communication measurement}
To analyze amortized communication complexity we assume that each
message contains a \emph{block} of transactions, and we say that a
transaction in a message $m$ is \emph{ordered} when all honest parties
$a\_deliver$ $m$.
We measure \emph{communication complexity} as the total number of bits
sent by honest processes to order a single transaction.
To be able to measure the asynchronous running time we
follow~\cite{canetti1993fast} and define a \emph{time unit} for every
execution $r$ to be the maximum time delay of all messages among
correct processes in $r$.
We measure \emph{time complexity} as the expected number of time units
it takes for a correct party to deliver $O(n)$ values proposed by
different correct processes starting from any point in the execution.

\section{DAG Abstraction} \label{sec:DAGconstruction}

\begin{figure*}[t]
	\centering
	\includegraphics[width=\textwidth]{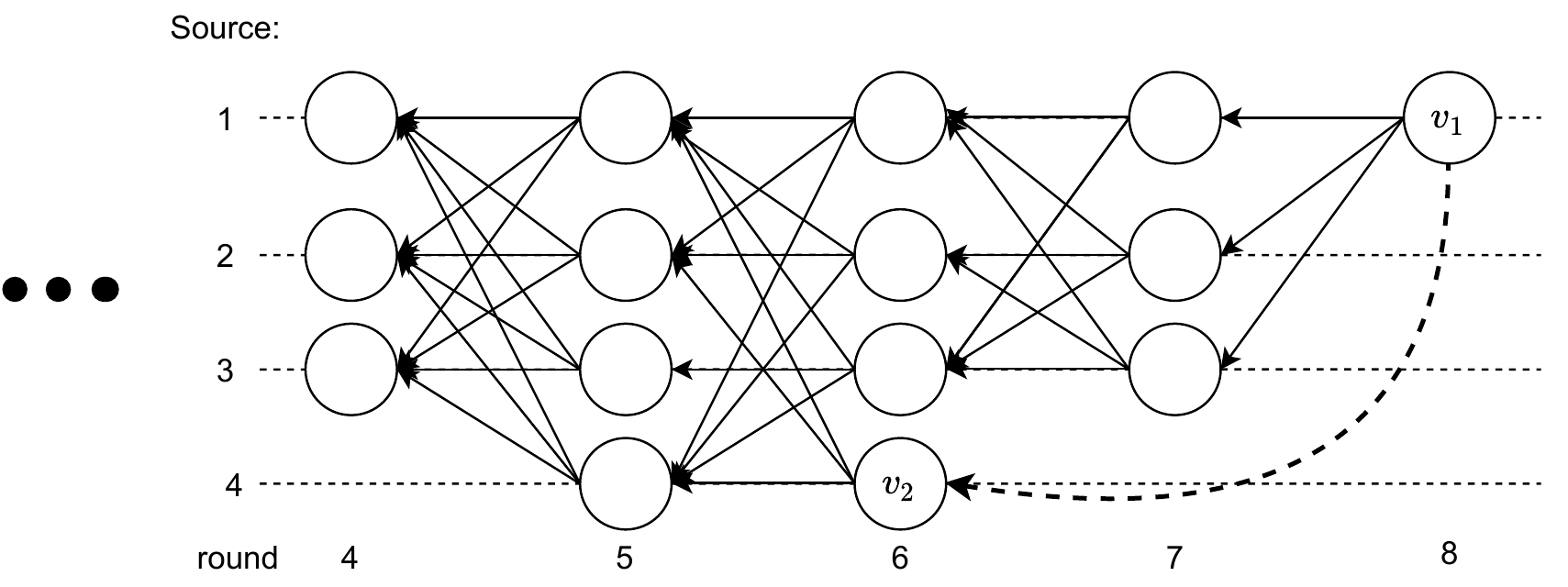}
	\caption{\textbf{Illustration of $\mathbf{DAG_1}$, i.e., the DAG at process 1,
	out of a total of four processes.}
On each horizontal dotted line are the vertices from a single source,
e.g, the bottom line shows the vertices delivered from process 4.
Each vertical column of vertices is a single round.
Each completed round has at least $2f+1 = 3$ vertices.
Each vertex in the DAG has at least $2f+1$ strong edges to vertices
from the previous round shown as black solid arrows.
Each vertex can also have weak edges to vertices in case there is
no other path in the DAG to the vertex.
E.g., $v_1$ in the illustration has a weak edge to $v_2$, shown as a
dotted arrow to $v_2$.}
\label{fig:DAGillustration}
\end{figure*}

\begin{algorithm*}[t]
    \caption{Data structures and basic utilities for process $p_i$}
    \label{alg:dataStructures}
    \small
    \begin{algorithmic}[1]
        \Statex \textbf{Local variables:}
        \StateX struct $\textit{vertex } v$: \Comment{The struct of a vertex in the DAG}
        \StateXX $v.\textit{round}$ - the round of $v$ in the DAG
        \StateXX $v.\textit{source}$ - the process that broadcast $v$
        \StateXX $v.\textit{block}$ - a block of transactions
        \StateXX $v.\textit{strongEdges}$ - a set of vertices in
        $v.\textit{round}-1$ that represent \emph{strong} edges 
        \StateXX $v.\textit{weakEdges}$ - a set of vertices in rounds $<
        v.\textit{round}-1$ that represent \emph{weak} edges 
        \StateX $DAG[]$ - An array of sets of vertices, initially:
            \StateXX $DAG_i[0] \gets$ predefined hardcoded set of $2f+1$
            vertices 
            \StateXX  $\forall j \geq 1 \colon DAG_i[j] \gets \{\}$ 
        \StateX $\textit{blocksToPropose}$ - A queue, initially empty,
        $p_i$ enqueues valid blocks of transactions from clients
        
        \vspace{0.5em}
        \Procedure{\textit{path}}{$v,u$} \Comment{Check if exists a path consisting of strong and weak edges in the DAG}
        \State \Return exists a sequence of $k \in \mathbb{N}$,
        vertices $v_1,v_2,\ldots,v_k$  s.t.\
        \StateXX $v_1 = v $, $v_k = u$, and $\forall i \in [2..k]
        \colon v_i \in \bigcup_{r \geq 1}
        DAG_i[r] \wedge (v_i \in v_{i-1}.\textit{weakEdges} \cup
        v_{i-1}.\textit{strongEdges})$
        \EndProcedure
        
        \vspace{0.5em}
        \Procedure{\textit{strong\_path}}{$v,u$} \Comment{Check if exists a path consisting of only strong edges in the DAG}
        \State \Return exists a sequence of $k \in \mathbb{N}$,
        vertices $v_1,v_2,\ldots,v_k$  s.t.\
        \StateXX $v_1 = v $, $v_k = u$, and $\forall i \in [2..k]
        \colon v_i \in \bigcup_{r \geq 1}
        DAG_i[r] \wedge v_i \in v_{i-1}.\textit{strongEdges}$

        \EndProcedure
        
        \alglinenoNew{counter}
        \alglinenoPush{counter}

    \end{algorithmic}
\end{algorithm*}

\begin{algorithm*}[t]
	\caption{\textbf{DAG Construction}, pseudocode for process $p_i$}
	\label{alg:DAGabstraction}
	\small
	\begin{algorithmic}[1]
	\alglinenoPop{counter}

		\Statex \textbf{Local variables:}
		\StateX $r \gets 0$ \Comment{round number}
		\StateX $\textit{buffer} \gets \{\}$

		\vspace{0.5em}
		
		\While{True}
		\For{$v \in \textit{buffer} \colon v.\textit{round} \leq r$} \label{alg:DAG:goThroughBuffer}
		  
		  \If{$\forall v' \in v.\textit{strongEdges} \cup v.\textit{weakEdges} \colon v' \in \bigcup_{k \geq 1} DAG[k]$} \Comment{We have $v$'s predecessors} \label{alg:DAGabstraction:prevVerticesIf}
		  
		      \State $DAG[v.\textit{round}] \gets DAG[v.\textit{round}] \cup
		      \left\{ v \right\}$ \label{alg:DAGabstraction:addVertexToDAG}
		      \State $\textit{buffer} \gets \textit{buffer} \setminus \{ v \}$ 
		  
		  \EndIf
		\EndFor

		\If {$\left| DAG[r] \right| \geq 2f + 1$} \Comment{Start a new round} \label{alg:DAG:checkNewRound}
		\If{$r \bmod 4 = 0$ } \Comment{If a new wave is complete}
		  \State $\textit{wave\_ready}(r/4)$ \label{alg:DAGabstraction:waveReady}
		  \Comment{Signal to \Cref{alg:SMROnDag} that a new wave is complete}
		\EndIf
		\State $r \gets r+1$ 
		\State $v \gets \textit{create\_new\_vertex}(r)$
		\State $\textit{r\_bcast}_i(v,r)$ \label{alg:DAGabstraction:RBcast}
		\EndIf
		\EndWhile
		
		\vspace{0.5em}
		\Procedure{\textit{create\_new\_vertex}}{round} \label{alg:DAG:createNewVertex}
		\State \textbf{wait until}
		$\neg$\textit{blocksToPropose}.\text{empty}()
		\Comment{atomic broadcast blocks are enqueued
        (\Cref{alg:SMR:enqueue})} 
		\State $v.\textit{block} \gets
		\textit{blocksToPropose}.\text{dequeue}()$
		\Comment{We assume each process atomically broadcast infinitely
		many blocks} \label{alg:SMR:dequeue} \State $v.\textit{strongEdges}
		\gets DAG[\textit{round} - 1]$ \State $\textit{set\_weak\_edges}(v,\textit{round})$ \State \textbf{return} $v$
		\EndProcedure
		
		\vspace{0.5em}
        \Upon{$\textit{r\_deliver}_i(v, \textit{round},p_k)$}
        \Comment{The deliver output from the reliable broadcast} \label{alg:DAGabstraction:reliableDeliver} \State $v.\textit{source} \gets p_k$
        \State $v.\textit{round} \gets \textit{round}$
        \If{$\left| v.\textit{strongEdges} \right|
        \geq 2f+1$} \label{alg:DAGabstraction:reliableIf}

            \State $\textit{buffer} \gets \textit{buffer} \cup \left\{ v \right\}$      
        
        \EndIf
        
        \EndUpon
		
% 		\vspace{0.5em}
% 		\Procedure{\textit{is\_valid\_vertex}}{$v$} \label{alg:DAG:isValidVertex}
% 		\State \Return {$\left| v.\textit{strongEdges} \right|
%         \geq 2f+1 \wedge \forall tx \in v.\textit{block} \colon
% 		\textit{validate}(tx)$} 
% 		\Comment{\parbox[t]{.32\linewidth}{$\textit{validate}(tx)$ checks
% 		external validity}}
% 		\EndProcedure
		
		\vspace{0.5em}
		\Procedure{\textit{set\_weak\_edges}}{$v, \textit{round}$} \Comment{Add weak
		edges to orphan vertices} \label{alg:DAG:getWeakEdges} 
		\State $v.\textit{weakEdges} \gets \{\}$
		\For{$r=\textit{round}-2$ down to 1}
		\For{\textbf{every} $u \in DAG_i[r]$ s.t. $\neg \textit{path}(v,u) $}
		\State $v.\textit{weakEdges} \gets v.\textit{weakEdges} \cup \{u \}$
		\EndFor
		\EndFor
		\EndProcedure
		
		\alglinenoPush{counter}
	\end{algorithmic}
\end{algorithm*}

Our BAB protocol, \sys, is based on a Directed
Acyclic Graph (DAG) abstraction, which represents the communication
layer of the processes.
In a nutshell, each vertex in the DAG represents a reliable broadcast
message from a process, and each message contains, among other data,
references to previously broadcast vertices.
Those references are the edges of the DAG.
Each correct process maintains a copy of the DAG as it perceives it. 
Different correct processes might observe different states of the DAG
during different times of the run, but reliable broadcast prevents
equivocation and guarantees that all correct processes eventually
deliver the same messages, so their views of the DAG eventually
converge.

For each process $p_i$, denote $p_i$'s local view of the
DAG as $DAG_i$, which is stored as an array $DAG_i[]$.
As we shortly explain, each vertex in the DAG is associated with a
unique round number and a source (its generating process).
At any given time, $DAG_i[r]$ for $r \in \mathbb{N}$ is the set of all the
vertices associated with round $r$ that $p_i$ is aware of.
Each round has at most $n$ vertices, each with a different source.
Due to the reliable broadcast, no process can generate two
vertices in the same round.
 
Each vertex $v$ in a round $r$ has two sets of outgoing edges: 
a set of at least $2f+1$ \emph{strong edges} and a set of up to $f$
\emph{weak edges}.
Strong edges point to vertices in round $r-1$ and weak edges point
to vertices in rounds $r' < r -1 $ such that otherwise there is no path
from $v$ to them.
As explained in detail in~\Cref{sec:VAB}, strong edges are used for
agreement and weak edges make sure we eventually include all
vertices in the total order, to satisfy BAB's validity property.

The data types and variables for process $p_i$ are specified in
\Cref{alg:dataStructures} and the DAG construction is specified in
\Cref{alg:DAGabstraction}.
A vertex $v$ is a struct that holds a round number $r$, a source which
is the process that created $v$, a block of valid transactions
that was previously $a\_bcast$ by the upper BAB protocol,
strong edges to at least $2f+1$ vertices in round $r-1$, and weak edges
to vertices in rounds $r' < r-1$.
Vertices in the DAG are reliably broadcast
(\Cref{alg:DAGabstraction:RBcast}), and when the reliable broadcast
layer delivers a vertex $v$
(\Cref{alg:DAGabstraction:reliableDeliver}), processes use the round
number $r$ and the source process which are available from the reliable
broadcast and add them to $v$.
Then, they verify that $v$ has strong edges to at least $2f+1$ vertices
from round $r-1$ and it to a buffer.
% If $v$ passes the check, it is added to a buffer.

Each process $p_i$ continuously goes through its buffer to check if
there is a vertex $v$ in it that can be added to its $DAG_i$
(\Cref{alg:DAG:goThroughBuffer}).
A vertex $v$ can be added to the DAG once the DAG
contains all the vertices that $v$ has a strong or weak edge to
(\Cref{alg:DAGabstraction:prevVerticesIf}).
When $p_i$ has at least $2f+1$ vertices in the current round, it
moves to the next round (\Cref{alg:DAG:checkNewRound}) by
creating and reliably broadcasting a new vertex $v'$. 
The new vertex $v'$ in round $r$ includes a block of transactions $b$
for which $p_i$ previously invoked $a\_bcast(b,r)$ (we assume each
process atomically broadcast infinitely many blocks), strong edges to
the vertices in $DAG_i[r]$ (\Cref{alg:DAG:createNewVertex}), and weak
edges to any vertices with no path from $v'$ to them
(\Cref{alg:DAG:getWeakEdges}).
Note that a vertex might be delivered at $p_i$'s DAG after $p_i$ has
moved to a later round.
In this case, the vertex is still added to the DAG, but $p_i$'s
vertices do not include strong edges to it.
Weak edges are possible.
As noted, the weak edges are used to ensure the BAB's Validity
property.
An example of our DAG construction is illustrated in
\Cref{fig:DAGillustration}.

\section{\sys: DAG-Based Asynchronous BAB Protocol}
\label{sec:VAB}

\begin{algorithm*}[t]
	\caption{\textbf{\sys: Byzantine Atomic Broadcast based on DAG.} Pseudocode for process $p_i$}
	\label{alg:SMROnDag}
	\small
	\begin{algorithmic}[1]
	\alglinenoPop{counter} 

		\vspace{0.5em}
		\Statex \textbf{Local Variables:}
%         \StateX $\textit{slot} \gets 1$
		\StateX $\textit{decidedWave} \gets 0$
		\StateX $\textit{deliveredVertices} \gets \{\}$
		\StateX $\textit{leadersStack} \gets$ initialize empty stack with isEmpty(), push(), and pop() functions 
		\vspace{0.5em}
		
		\Upon{$\textit{a\_bcast}_i(b,r)$} \Comment{Correct processes call this procedure with sequential round $r$ numbers, starting at $1$} \label{alg:SMR:enqueue}
		
		\State $\textit{blocksToPropose}.\text{enqueue}(b)$
		\label{line:propose}
		\Comment{pushes a block of transactions to
		Alg~\ref{alg:DAGabstraction}}
		
		\EndUpon
		\vspace{0.5em}
		
		\Upon{$\textit{wave\_ready}(w)$} \Comment{Signal from the DAG layer
		that a new wave is completed (\Cref{alg:DAGabstraction:waveReady})}
		\label{alg:SMR:waveCompletion} 
		\State $v \gets \textit{get\_wave\_vertex\_leader}(w)$ \label{alg:SMR:getWaveLeader}
		\If{$v = \bot \vee \left|  \left\{ v' \in
		DAG_i[\textit{round}(w,4)]
		\colon
		\textit{strong\_path}(v',v) \right\} \right| < 2f+1$ } 
		\Comment{No commit}
		\label{alg:SMR:commitrule}
		  \State \Return
		\EndIf

        \State $\textit{leadersStack}.\text{push}(v)$ \label{alg:SMR:stackPush1}
        \For{wave $w'$ from $w - 1$ down to
        $\textit{decidedWave} + 1$} \label{alg:SMR:orderLadersForLoop}
        \State $v' \gets \textit{get\_wave\_vertex\_leader}(w')$
        \If{$v' \neq \bot \wedge \textit{strong\_path}(v,v')$}
            \label{alg:DAGabstraction:checkpath}
            \State $\textit{leadersStack}.\text{push}(v')$ \label{alg:SMR:stackPush2}
            \State $v \gets v'$
            \label{alg:SMR:orderLadersForLoopEnd}
            
        \EndIf
        \EndFor
        \State $\textit{decidedWave} \gets w$ \label{alg:SMR:decidedWaveUpdate}
		  \State $\textit{order\_vertices}(\textit{leadersStack})$
		
		\EndUpon
		\vspace{0.5em}

		\Procedure{\textit{get\_wave\_vertex\_leader}}{$w$} \label{alg:SMR:vertexFromCoin}
            \State $p_j \gets \textit{choose\_leader}_i(w)$
            \If{$\exists v\in DAG[\textit{round}(w,1)]$ s.t.\ $v.source = p_j$}
              \State \Return $v$ \Comment{There can only be one such vertex}
          \EndIf 
          \State \Return $\bot$ 
        \EndProcedure
        \vspace{0.5em}

		\vspace{0.5em}
		\Procedure{$\textit{order\_vertices}$}{\textit{leadersStack}}
		\label{alg:SMR:decideValuesProcedure}
		\While{$\neg \textit{leadersStack}.\text{isEmpty}()$} 
		\State $v \gets \textit{leadersStack}.\text{pop}()$ \label{alg:SMR:stackPop}
		  \State \textit{verticesToDeliver} $\gets \{v' \in \bigcup_{r > 0}
        DAG_i[r] \mid
		  path(v,v') \wedge v' \not\in \textit{deliveredVertices}\}$
		  \For{$\textbf{every} ~v' \in \textit{verticesToDeliver}$ in some deterministic
		  order}
	  		  \State \textbf{output}
	  		  $\textit{a\_deliver}_i(v'.\textit{block},v'.\textit{round},
	  		  v'.\textit{source})$
	  		  \label{alg:SMR:decide}
% 		      \State $\textit{slot} \gets \textit{slot} + 1$
		      \State $\textit{deliveredVertices} \gets \textit{deliveredVertices} \cup \{v'\}$
		  \EndFor
		\EndWhile
		\EndProcedure

	\alglinenoPush{counter}	
	\end{algorithmic}
\end{algorithm*}

\begin{figure*}[t]
	\centering
	\includegraphics[width=\textwidth]{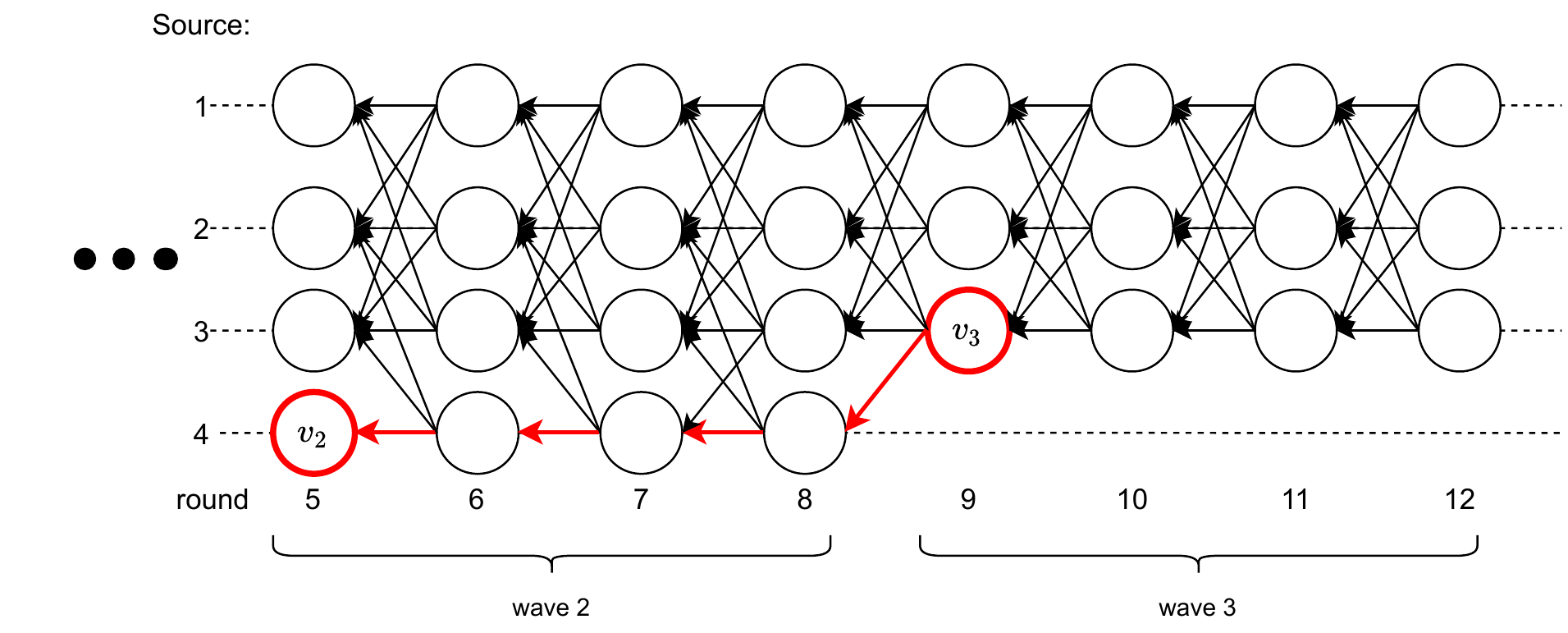}
	\caption{\textbf{Illustration of $\mathbf{DAG_1}$.}
	The highlighted vertices $v_2$ and $v_3$ are
	the leaders of waves 2 and 3, respectively.
	The commit rule is not met in wave 2 since there are less than $2f+1$
	vertices in round 8 with a strong path to $v_2$.
	However, the commit rule is met
	in wave 3 since there are $2f+1$ vertices in round 12 with a strong
	path to $v_3$.
	Since there is a strong path from $v_3$ to $v_2$ (highlighted), $p_i$ commits $v_2$
	before $v_3$ in wave 3.}
	\label{fig:decisionRule}
\end{figure*}

In this section, we describe the \sys protocol, by equipping the DAG
from the previous section with a global perfect coin\footnote{A
possible implementation of the coin using threshold signatures is
described in \Cref{sec:model}. The coin can be easily implemented as
part of the DAG itself by having each process send its share of the
threshold signature when reliably broadcasting a vertex.} and show how
the DAG and the coin can be used to construct a locally-computed
protocol for the BAB problem.
That is, given our DAG and a perfect coin, \sys
does not require any extra communication among the processes.
Instead, each process $p_i$ observes its local $DAG_i$ and deduces
which blocks of transactions to deliver and in what order.
The protocol is detailed in \Cref{alg:SMROnDag}.
Below we give a high-level intuition as well as a detailed description
of the protocol.
Formal correctness proofs and complexity
analyzes are given in~\Cref{sec:analysis}.
%  in \Cref{sec:analysis} we analyze complexity and prove the correctness
%  of the full protocol.

When an $\textit{a\_bcast}_i(b,r)$ is invoked, $p_i$ simply pushes $b$
to the DAG layer (\cref{line:propose}), which in turn includes it in the
$r^{th}$ vertex it reliably broadcasts.
%%%Explain wave
To interpret the DAG, each process $p_i$ divides its local $DAG_i$ into
waves, where each wave consists of 4 consecutive rounds. 
For example, $p_i$'s first wave consists of $DAG_i[1], DAG_i[2],
DAG_i[3]$, and $DAG_i[4]$.
Formally, the $k^{th}$ round of wave $w$, where $k \in [1..4], w \in \mathbb{N}$, is defined as $\textit{round}(w,k) \triangleq 4(w-1)+k$.
We also say that a process $p_i$ \emph{completes round $r$} 
once $DAG_i[r]$ has at least $2f+1$ vertices, and a process
\emph{completes wave $w$} once the process completes $\textit{round}(w,4)$.

In a nutshell, the idea is to interpret the DAG as a wave-by-wave
protocol and try to \emph{commit} a randomly chosen single leader
vertex in every wave.
Once the sequence of leaders is determined,
processes $a\_deliver$ all the blocks included in their causal
histories (in vertices that have paths from the leaders in the DAG).
While reading the high-level
description below, bear in mind that due to the reliable broadcast, Byzantine processes cannot equivocate, so two
correct processes cannot have different vertices with the same source in
the same round, leading to eventually consistent DAGs among all correct processes.

%%%Explain commit rule (cite Lemma informally)
When wave $w$ completes (\Cref{alg:SMR:waveCompletion}), we use the global perfect coin to
retrospectively elect some process and consider its vertex in the wave's first round as the leader of wave $w$ (\Cref{alg:SMR:getWaveLeader}).
The goal of the protocol is to commit this leader, provided that it has been observed by sufficiently many processes in the wave.
Note that since we advance rounds as soon as we deliver $2f+1$ of the $3f+1$ potential vertices, a process
$p_i$ might not have $w$'s leader in its local $DAG_i$ when
it completes $w$.
In this case, $p_i$ completes $w$ without committing any vertex and simply proceeds to the next wave.
Note, however, that some other correct process might have $w$'s leader in its local DAG and commit it in the same wave.
Thus, we need to make sure that if one correct process
commits the wave vertex leader $v$, then all the other correct processes will commit $v$ later.
To this end, we use standard quorum intersection.
Process $p_i$ commits the wave $w$ vertex leader $v$ if:
\[
\left|  \left\{ v' \in DAG_i[\textit{round}(w,4)] \colon
\textit{strong\_path}(v',v) \right\} \right| \geq
2f+1
\text{ (\Cref{alg:SMR:commitrule})}.
\]

In addition, if $p_i$ commits vertex $v$ in wave $w$ and there is a
strong path from $v$ to $v’$ such that $v’$ is an uncommitted leader
vertex in a wave $w' <w$, then $p_i$ commits $v'$ in $w$ as well.
The leaders committed in the same wave are ordered by their round
numbers, so that leaders of earlier waves are ordered before those of
later ones, meaning $v'$ is ordered before $v$
(Lines~\ref{alg:SMR:orderLadersForLoop}-\ref{alg:SMR:orderLadersForLoopEnd}).

The next lemma, which is proven in Section~\ref{sec:analysis}, shows that our
commit rule guarantees that if a correct process commits a wave leader
vertex $v$ in some wave, then all wave vertex leaders in later waves in
the local DAGs of all correct processes have a strong path to $v$, ensuring the agreement property.

%\begin{restatable}{lemma}{agreementLemma}
%\label{lem:commit}
%
%If some process $p_i$ commits the leader vertex $v$ of a wave $w$, then
%for every leader vertex $u$ of a wave $w' > w$ and for every process
%$p_j$, if $u \in DAG_j[\textit{round}(w',1)]$, then
%$\textit{strong\_path}(u,v)$ returns true in wave $w'$.
%\end{restatable}

%\restate{agreementLemma} \label{lem:commit}

% \begin{restatable}{lem}{quorum}
% \label{lem:quorum}
% If an honest validator commits a leader block $b$ in an instance $i$,
% then any leader block $b'$ committed by any honest validator $v$ in
% future instances have a path to $b$ in $v$'s local DAG.
% \end{restatable}

\begin{restatable}{lemma}{commitPath}
\label{claim:commitPath}\label{lem:commit}
	If some process $p_i$ commits the leader vertex $v$ of a wave $w$, then
	for every leader vertex $u$ of a wave $w' > w$ and for every process
	$p_j$, if $u \in DAG_j[\textit{round}(w',1)]$, then
	$\textit{strong\_path}(u,v)$ returns true in wave $w'$.
\end{restatable}

%%%Explain common core and how it gives constant latency (cite Lemma
%%%informally)
We show below how we leverage the above lemma to satisfy the
total order property, but first, we give an intuition for liveness,
i.e., the validity and agreement properties.
Our protocol achieves progress in a constant number of waves, in
expectation, by guaranteeing that for every wave, the probability for
every correct process to commit the wave leader is at least $2/3$.
To ensure this, we borrow the technique from the \emph{common-core
abstraction}~\cite{canetti1996studies}, which guarantees that after
three rounds of all-to-all sending and collecting accumulated sets of
values, all correct processes have at least $2f+1$ common values. The
set of these values is referred to as the common-core.
In respect to our DAG, we prove in Section~\ref{sec:analysis}
the following lemma:

\begin{restatable}{lemma}{commonCore}
\label{lem:commonCore}
%\begin{claim} \label{claim:analysisCommonCore}
	Let $p_i$ be a correct process that completes wave $w$.
	Then there is a set $V \subseteq DAG_i[\textit{round}(w,1)]$ and a set $U \subseteq DAG_i[\textit{round}(w,4)]$ s.t. $|V| \geq 2f+1$, $|U| \geq 2f+1$ and $\forall v \in V,\forall u \in U \colon
	\textit{strong\_path}(u,v)$.
\end{restatable}

Note that by the commit rule, if the leader of a wave $w$ belongs to
the set $V$ (from the lemma statement), then $p_i$ commits the leader
once it completes $w$.
So to deal with an adversary that totally controls the network, 
parties flip the global coin only after they complete $w$
(\Cref{alg:SMR:getWaveLeader}).
Therefore, by the coin's unpredictability property, the probability of
the adversary to guess the wave's leader before the point after which
it cannot manipulate the set $V$ is less than $\frac{1}{n} +
\epsilon$.
Thus, with a probability of at least $2/3-\epsilon$, $w$'s leader
is in the set $V$ and $p_i$ commits it.
Thus, in expectation, correct processes commit every $3/2$ waves.

% Explain how we order leaders.
To satisfy total order, we leverage the property proven in
\Cref{lem:commit} to make sure all processes commit the same waves' leaders.
Once we find a leader to commit in a wave $w$ we check if it is
possible that some process committed a wave in between $w$ and the
previous wave we committed, let it be $w'$. 
We do this iteratively in
Lines~\ref{alg:SMR:orderLadersForLoop}-\ref{alg:SMR:orderLadersForLoopEnd},
we first check if it is possible that some process committed the leader
of $w-1$.
We do it by checking if there is a strong path
from the leader of wave $w$ to the leader of wave $w-1$ in our local
DAG (\Cref{alg:DAGabstraction:checkpath}).
If no such path exists, by Lemma~\ref{lem:commit}, no correct process will
ever commit $w-1$.
Otherwise, we choose to commit $w-1$ before $w$.
Now, if such a path indeed exists, we recursively check if it is
possible that some process committed a wave in between $w - 1$ and $w'$.
Otherwise, if no such path exists, we check if there is a path from the
leader of wave $w$ to the leader of wave $w-2$ and continue in the
same way. 
The recursion ends once we reach a wave that we previously committed,
$w'$ in our example.
An illustration of this process is given in~\Cref{fig:decisionRule}.

%Explain how we order causal history. Talk about Eventual Fairness and
%weak edges.
Since vertices are reliably broadcast and since we never add a vertex $v$ to the DAG before we add all the
vertices $v$ points to with strong or weak edges, two correct processes
always have the same causal history for any vertex they both have in their DAGs.
Therefore, once we agree on a sequence of leaders,
all that is left to do is to order the causal histories of the leaders in some deterministic order.
To this end, we go through the waves' committed leaders one-by-one and
$a\_deliver$, in some deterministic order, all the transaction blocks
in their causal histories that we did not previously deliver
(procedure
\textit{order\_vertices} in \Cref{alg:SMR:decideValuesProcedure}).
The causal history of a wave leader vertex $v$ in $DAG_i$ is the set
$\{ u \in DAG_i \mid \textit{path}(v,u) \}$.

The purpose of the weak edges is to satisfy the Validity
property.
Recall that strong edges might not point to all vertices from the
previous round in the DAG because we might advance the round before we
deliver all the broadcasts of that round (we advance the round once at least $2f+1$ vertices are added to the DAG).
Therefore, without the weak edges, slow processes may not be able to
get vertices from higher rounds to point to theirs.
So to satisfy Validity, each correct process, when creating a new
vertex, adds weak edges to all vertices in its local DAG to
which it otherwise does not point.
\section{Analysis}
\label{sec:analysis}

In \Cref{sub:correctness} we prove the correctness of \sys, and in
\Cref{sub:complexity} we analyze the communication and time complexity.

\subsection{Correctness}
\label{sub:correctness}

We show that \sys satisfies the properties of the BAB problem, as
defined in \Cref{sec:problemDef}.

%\subsubsection{Safety}
%\subsubsection{Integrity}

\begin{proposition}
\sys satisfies the integrity property of the BAB problem.
\end{proposition}

\begin{proof}

By the code (\Cref{alg:SMR:decide}), if a correct process $p_i$ outputs
$a\_deliver(b,r,p_k)$, then there is a vertex $v$ in $DAG_i$ s.t.\
$(b,r,p_k) = (v.block, v.round, v.source)$.
Integrity follows from the fact that all vertices are
reliably broadcast, and thus by integrity property of the reliable
broadcast there are no two different vertices $u,u'$ in $DAG_i$ s.t. $u.round =
u'.round$ and $u.source = u'.source$.
\end{proof}

% \begin{proof}
% When a process atomically delivers a block (\Cref{alg:SMR:decide}) it does so to a block stored in a vertex in the DAG.
% The vertices in the DAG are reliably broadcast, and after atomically delivered, they are stored in the \textit{deliveredVertices} variable.
% Therefore, each delivered vertex is delivered at most once. 
% %When a correct process $p_i$ decides on a block in a slot number $\ell$ (\Cref{alg:SMR:decide}) it increases by one the \textit{slot} variable in the next line, so that next time $p_i$ decides on a block, it will be in slot number $\ell+1$, ensuring that $p_i$ does not decide on more than one block in a single slot number.
% \end{proof}

%\subsubsection{External validity}

%\begin{lemma}
%\sys achieves the external validity property of the BAB problem.
%\end{lemma}
%\begin{proof}
%A vertex $v$ is added to the $DAG_i$ of process $p_i$ only after $p_i$ checks if all transactions in $v.\textit{block}$ pass the external validity predicate (\Cref{alg:DAG:isValidVertex}).
%Since $p_i$ decides only on blocks in vertices that are in its DAG, then any decision made is on an externally valid block of transactions.
%\end{proof}

\begin{claim} \label{lem:causalHistory}
	When a correct process $p_i$ adds a vertex $v$ to its $DAG_i$ (\Cref{alg:DAGabstraction:addVertexToDAG}), all of $v$'s causal history is already in $DAG_i$.
\end{claim}

\begin{proof}
We prove this claim by induction on the execution of every correct process $p_i$.
Denote by $v_k$ the $k$-th vertex that $p_i$ adds to $DAG_i$.
We show that for every $k \in \mathbb{N}$, after $v_k$ is added to the DAG, the causal histories of all the vertices in the set $\{v_1, \ldots, v_k\}$, and in particular $v_k$, are in $DAG_i$.

In the base step of the induction, there are no vertices in the DAG, and the property vacuously holds.
Next, assume that after $v_k$ is added to the DAG at process $p_i$, all the causal histories of all the vertices in the set $V = \{v_1, \ldots, v_k\}$ are already in $DAG_i$.

For $v_{k+1}$ to be added to the DAG at process $p_i$, its strong and
weak edges must reference vertices that are already in $DAG_i$
(\Cref{alg:DAGabstraction:prevVerticesIf}), i.e., $v_{k+1}$'s edges are
only to vertices in $V$.
Since all the vertices in $V$ already have their causal histories in
the DAG, when $v_{k+1}$ is added to the DAG, its causal history is in
the DAG as well, and we are done.
\end{proof}

\begin{claim} \label{claim:consistentDAG}
If a correct process $p_i$ adds a vertex $v$ to its $DAG_i$, then eventually all correct processes add $v$ to their DAG.
\end{claim}

\begin{proof}
By induction on rounds, for process $p_i$ to add a vertex $v$ in round $r$ to its $DAG_i$, first $v$ needs to be delivered to $p_i$ by the reliable broadcast layer (\Cref{alg:DAGabstraction:reliableDeliver}), and by the agreement of the reliable broadcast, $v$ will be eventually delivered to all other correct processes.

Next, $v$ has to be added to the $\textit{buffer}$ variable at $p_i$, and this is done if the process who broadcast $v$ added the correct $v.\textit{source}$ and $v.\textit{round}$ which are verified through the guarantees of the reliable broadcast layer (\Cref{alg:DAGabstraction:reliableIf}).
Therefore these checks will also pass at any other correct process when $v$ is delivered to it.
Finally, $p_i$ checks that the vertex has at least $2f+1$ strong edges to vertices in round $v.\textit{round}-1$.
If $v$ passes this check in $p_i$ then it will pass these two checks at any other correct process, since these checks are computed locally based on $v$'s fields ($v.\textit{block}$ and $v.\textit{strongEdges}$).

Lastly, after $v$ is added to the \textit{buffer}, for $p_i$ to add $v$ to its $DAG_i$, $p_i$ also checks that it has all the vertices that $v$ is referencing to (in $V = v.\textit{strongEdges} \cup v.\textit{weakEdges}$) in its $DAG_i$ as well (\Cref{alg:DAGabstraction:prevVerticesIf}).
By the induction assumption, all correct processes' DAGs contain the same vertices in rounds $<r$.

Thus, this ensures that any vertex $v$ that appears in any round at $DAG_i$ of some correct process, will eventually also appear in $DAG_j$ of every other correct process $p_j$.
\end{proof}

%\begin{lemma} \label{lem:identicalPaths}
%If at some process $p_i$ there are two vertices $v,u \in DAG_i$ s.t. $v \xrightarrow[i]{}^* u$, then if $v \in DAG_j$ for some other correct process $p_j$, then $u \in DAG_j$ as well and $v \xrightarrow[j]{}^* u$.
%\end{lemma}
%
%\begin{proof}
%By~\Cref{claim:consistentDAG}, if $v,u$ appear in $DAG_i$, they will also eventually appear in $DAG_j$, and also any other vertex along the path between $v$ and $u$ in $DAG_i$.
%Since the edges of those vertices is also consistent among the two graphs, the path  $v \xrightarrow[i]{}^* u$, will eventually appear in $DAG_j$ through the exact same vertices.
%\end{proof}

\begin{claim} \label{claim:quorumIntersection}
If for some correct process $p_i$ there is a round $r$ with a set $V$ of at least $2f+1$ vertices in $DAG_i[r]$ s.t. $\forall v \in V \colon \textit{strong\_path}(v,u)$ to some vertex $u \in DAG_i$, then every other process $p_j$ that completes round $r$ has a set $V' \subseteq DAG_j[r]$ s.t. $|V'| \geq f+1$ and $\forall v' \in V' \colon \textit{strong\_path}(v',u)$.
\end{claim}

\begin{proof}
%Since there is a strong path between every $v \in V$ and $u$ and by \Cref{lem:causalHistory}, by the time $v$ is in the $DAG_i$ of $p_i$, $u$ is in $DAG_i$ as well.
Let $V' = V \cap DAG_j[r]$.
Round $r$ is complete for $p_i$ and $p_j$ when their DAGs have at least $2f+1$ vertices.
Therefore, when $p_i$ and $p_j$ complete round $r$, $|V'| \geq f+1$ by a standard quorum intersection of $2f+1$ out of $3f+1$ possible vertices of round $r$ (due to the reliable broadcast, Byzantine processes cannot equivocate).
Since every $v' \in V'$ is already in $DAG_j$ when $p_j$ completes round $r$, then $u$ is in $DAG_j$ by $t$ as well (by \Cref{lem:causalHistory}), and there is a strong path between every $v' \in V'$ to $u$ in $DAG_j$.
\end{proof}

For the next part, we say a vertex $v$ is a wave $w$ vertex leader if $v$ is the return value of the \textit{get\_wave\_vertex\_leader} procedure (\Cref{alg:SMR:vertexFromCoin}).
Next, we say a process \emph{commits} a wave leader vertex $v$ when $v$ is popped from the stack  (\Cref{alg:SMR:stackPop}).

\begin{claim} \label{claim:oneVertexLeader}
In every wave, at most one vertex $v$ can be a wave leader vertex for all correct processes.
\end{claim}

\begin{proof}
For a vertex $v$ to be a wave leader vertex in wave $w$ it has to be the return value from the \textit{get\_wave\_vertex\_leader} procedure (\Cref{alg:SMR:vertexFromCoin}).
The procedure gets the wave's chosen process $p_j$ by the global coin, and checks if the $DAG_i$ at process $p_i$ has the vertex $v$ from $p_j$ in the first round of wave $w$.
Due to the agreement property of the global perfect coin, the same process $p_j$ is chosen for all correct processes, and because of the agreement property of the reliable broadcast, Byzantine processes cannot equivocate.
\end{proof}

%\begin{claim}
%Let $P \subseteq \Pi$ be the set of correct processes that commit a vertex in wave $w$.
%Then all processes in $P$ commit the same vertex $v$.
%\end{claim}
%
%\begin{proof}
%By \Cref{claim:oneVertexLeader}, there is only one vertex in each wave that can be the wave leader vertex.
%Since all correct processes commit only the wave leader vertex, the result is immediate.
%\end{proof}

\begin{claim} \label{claim:increasingCommitOrder}
If a correct process $p_i$ commits  wave leader vertex $v_1$ in wave $w_1$ and after that $p_i$ commits vertex $v_2$ in wave $w_2$, then $w_1 < w_2$.
\end{claim}

\begin{proof}
A vertex is committed when it is popped from the stack (\Cref{alg:SMR:stackPop}).
Vertices are pushed to the stack in Lines \ref{alg:SMR:stackPush1} and \ref{alg:SMR:stackPush2}, which only happens in waves which vertices were not committed before, since the for loop goes down only to $\textit{decidedWave} + 1$ (\Cref{alg:SMR:orderLadersForLoop}), where \textit{decidedWave} is updated each time vertices are pushed to the stack to be the maximum wave in which vertices were committed (\Cref{alg:SMR:decidedWaveUpdate}).
This means that vertices are pushed to the stack in decreasing wave numbers.

Lastly, all the vertices in the stack are popped out and committed, and this is done in reverse order to the order that they were pushed to the stack, therefore, the wave numbers of committed waves are in an increasing order.
\end{proof}

% Next, we restate and prove \Cref{claim:commitPath}.
%  \begin{lemma} \label{claim:commitPath1}
%  	If some process $p_i$ commits the leader vertex $v$ of a wave $w$, then
%  	for every leader vertex $u$ of a wave $w' > w$ and for every process
%  	$p_j$, if $u \in DAG_j[\textit{round}(w',1)]$, then
%  	$\textit{strong\_path}(u,v)$ returns true in wave $w'$.
%  \end{lemma}

\commitPath*

\begin{proof}
Since vertex $v$ is committed by process $p_i$ in wave $w$, the commit rule is met, i.e., at the end of wave $w$ there are at least $2f+1$ vertices in $DAG_i[\textit{round}(w,4)]$ with a strong path to $v$.
By \Cref{claim:quorumIntersection}, every correct process $p_j$ (whether it committs $v$ in $w$ or not) has a set $V$ of at least $f+1$ vertices in $DAG_j[\textit{round}(w,4)]$ with a strong path to $v$.
Any future vertex $v'$ from waves $w' > w$, including $u$, will have a strong path to at least one vertex in $V$, resulting in a strong path between $u$ and $v$.
%This does  not matter if vertex $v'$ is created by a Byzantine process
%or not, since Byzantine processes cannot equivocate due to the reliable
%broadcast guarantees.
\end{proof}

\begin{proposition}
\label{lem:agreement}
\sys satisfies the total order property of the BAB problem.
\end{proposition}

\begin{proof}
By \Cref{claim:oneVertexLeader}, each wave has only one vertex that can be committed.
By \Cref{claim:increasingCommitOrder} every correct process commits vertices in an increasing wave number.
By \Cref{claim:commitPath}, if a correct process $p_i$ commits a vertex $v$, then there is a strong path to $v$ from any vertex $u$ in future waves that might be committed.
By combining all the claims, if two correct processes commit the same wave leader vertices, they do so in the same order.

Once a correct process commits a wave vertex leader $v$, it atomically delivers all of $v$'s causal history in some deterministic order, which is identical for all other correct processes.
By \Cref{lem:causalHistory}, when $v$ is committed, all of $v$'s causal history is already in the DAG.
Thus, since all correct processes commit the same wave leader vertices in the same order, and since those vertices have the same causal histories, all correct processes that deliver vertices, do so in the same order.
\end{proof}

\commonCore*

\begin{proof}
First, we show that there is a set $V$, $|V| \geq 2f+1$ s.t. when $p_i$ completes $\textit{round}(w,3)$ and broadcasts a new vertex $v_4$ in $\textit{round}(w,4)$, then $v_4$ has a strong path to all the vertices in $V$.

To this end, we use the common-core abstraction, that first appeared in~\cite{canetti1996studies}, and was adapted (and proven) for the Byzantine case in~\cite{dolev2016garbage}.
The model for this abstraction is identical to our model.
Each correct process $p_i$ has some input value $v_i$, and it returns a set $V_i$ of input values from different processes.
The guarantee of the common-core abstraction is that there is a subset $V$ of at least $2f+1$ values, s.t. for each correct process $V \subseteq V_i$, i.e., there is a common core of at least $2f+1$ input values that appear in the returned sets of all the correct processes that complete the common-core abstraction.

The algorithm to realize the common-core abstraction consists of three rounds of communication: in the first round, each process sends its input value $v_i$, and then waits for $2f+1$ input values from other processes (including itself).
Denote this first set at process $p_i$ as $F_i$.

In the second stage, each process sends its $F_i$ set and waits until it receives $2f+1$ $F_j$ sets from other processes (including itself).
When this stage ends, process $p_i$ creates the union of all the $F_j$ sets it received.
Denote this set of sets for process $p_i$ as $S_i$.
In the third and last stage, process $p_i$ sends the set $S_i$ it created and again waits to receive $2f+1$ $S_j$ sets from other processes (including itself). 
When this stage ends, process $p_i$ returns the union of all the $S_j$ sets, denoted $T_i$, as the output of the common-core abstraction.

We show that the first three rounds of a wave $w$ can be mapped exactly to the three stages of the common-core algorithm.
Denote $r_1, r_2,r_3,r_4$ as $\textit{round}(w,1),\textit{round}(w,2),\textit{round}(w,3),\textit{round}(w,4)$, respectively.
When a correct process $p_i$ adds the vertex $v_1$ created in $r_1$ to $DAG_i[r_1]$, by \Cref{claim:consistentDAG}, eventually all other correct processes add $v_1$ to their DAG, which can be mapped to $p_i$ sending its input value to all other processes in the common-core algorithm.
Next, $p_i$ moves to round $r_2$ once it has at least $2f+1$ vertices in $r_1$, which is mapped to $p_i$ waiting for $2f+1$ input values from different processes in the common-core algorithm.
When $p_i$ enters $r_2$ it broadcasts a vertex $v_2$ that references all the vertices it has in $r_1$, which is equivalent to $p_i$ sending $F_i$ at the beginning of the second stage of the common-core algorithm.
In a similar way, when $p_i$ completes $r_2$ and enters $r_3$, it broadcasts $v_3$ which references all the vertices it has in $r_2$, which is equivalent to sending $S_i$ (by \Cref{lem:causalHistory}, when $v_3$ is added to $DAG_j[r_3]$ of some correct process $p_j$, then all the vertices $p_i$ has in $DAG_i[r_1]$ with a strong path from $v_3$ are in the $DAG_j[r_1]$ as well).
To complete the mapping, when $p_i$ completes $r_3$ and broadcasts $v_4$ in round $r_4$, then $v_4$ has in its causal history the same values that would have been in $T_i$ in the equivalent common-core algorithm.

Note that since Byzantine processes cannot equivocate, and since every round in the DAG has at least $2f+1$ vertices, any vertex that $p_i$ adds to $DAG_i[r_4]$ has to reference at least $f+1$ vertices that $v_4$ also references, even vertices sent from Byzantine processes.
Thus, based on the common-core guarantee, there is a set $V \subset DAG_i[r_1]$ s.t. $|V| \geq 2f+1$ and $\forall v \in V \colon \textit{strong\_path}(v_4,v)$, and also this set $V$ appears in the DAG of any other correct process $p_j$ that completes round $r_3$.
Next, when $p_i$ completes wave $w$, i.e., when it completes round $r_4$, it has in $DAG_i[r_4]$ at least $2f+1$ vertices, and each of those vertices has a path to each of the vertices in $V$, which concludes the proof.
\end{proof}

\begin{claim} \label{claim:eventualCommit}
For every correct process $p_i$ and for every wave $w$, the
expected number of waves, starting from $w$, until the commit rule is
met is equal to or smaller than $3/2 + \epsilon$.

\end{claim}
\begin{proof}
By \Cref{lem:commonCore}, in each wave $w$, the probability that for a correct process $p_i$ the commit rule is met is at least $pr = (2f+1)/(3f+1) - \epsilon$.
The number of waves until the commit rule is met is geometrically distributed with a success probability of $pr$.
Thus, the expected number of waves is bounded by $3/2+\epsilon$ waves.
\end{proof}

\begin{proposition}
\sys guarantees the agreement property of the BAB problem.
\end{proposition}
\begin{proof}
If a correct process $p_i$ outputs $\textit{a\_deliver}_i(b,r,p_k)$ it means that $b$ is a block of some vertex $u$ that is in the causal history of some wave's $w$ leader vertex $v$ that, i.e., when process $p_i$ commits a wave vertex leader $v$, then $u$ is in $v$'s causal history.

By \Cref{claim:eventualCommit}, every other correct process $p_j$ that has not committed $v$ yet will eventually, with probability 1, have a wave $w' > w$ in which the commit rule is met. 
When $p_j$ commits $w'$, by the proved total order property, it will also commit $v$, and thus decide on all of $v$'s causal history in the same order, including vertex $u$.
\end{proof}

\begin{claim} \label{claim:eventualDelivery}
Every vertex that is broadcast by a correct process is eventually added to the DAG of all correct processes.
\end{claim}

\begin{proof}
We prove this by showing that for every correct process $p_i$ that broadcasts a vertex $v$, $v$ is eventually added to $DAG_i$, and by \Cref{claim:consistentDAG}, $v$ is eventually added to the DAG of all other correct processes.

When a correct process $p_i$ broadcasts a vertex $v$ (\Cref{alg:DAGabstraction:RBcast}) it broadcasts a valid vertex, i.e., a vertex that passes the external validity check, and that references vertices that are already in $DAG_i$.
Because of the validity property of the reliable broadcast, $o_i$ eventually delivers $v$ to itself, and when it does so, it adds $v$ to its own $DAG_i$.
Thus, as explained, by \Cref{claim:consistentDAG}, $v$ is eventually added to the DAGs of all other correct processes.
\end{proof}

\begin{proposition}
\label{lem:fairness}
	\sys guarantees the validity property of the BAB problem.
\end{proposition}
\begin{proof}
When a correct process $p_i$ calls \textit{a\_deliver} with some value, it is inserted into a queue (\Cref{line:propose}), and eventually will be included in a vertex $v$ created by $p_i$ (\Cref{alg:SMR:dequeue}).
Vertex $v$ is eventually reliably delivered to all the correct processes and added to their DAGs (\Cref{claim:eventualDelivery}).

When a correct process reliably broadcasts a new vertex $v$ in round $r$ it also makes sure that it has a path (either a strong path or path that includes weak edges) to all the vertices in rounds $r' < r$, and if not, it adds weak edges to $v$ that guarantee this (\Cref{alg:DAG:getWeakEdges}), therefore $v$ will eventually be included in the causal history of all correct processes.
Eventually, with probability 1, $v$ will be in
the causal history of a committed wave vertex leader, and therefore atomically delivered.
\end{proof}

% We proved that \sys achieves all the properties of the BAB problem.

\subsection{Communication and Time Complexity}
\label{sub:complexity}

We analyze \sys in terms of expected communication complexity and expected time complexity.

\paragraph{Communication complexity}
We analyze the communication complexity of \sys when instantiated
with Cachin and Tesero's~\cite{cachin2005asynchronous} information dispersal protocol.
A similar analysis can be made for other broadcast implementations as
well.
For clarity, in \Cref{sec:DAGconstruction}, we say that strongEdges and
weakEdges are sets of vertices.
However, in order to refer to a vertex it is enough to only store its source and
round number.\footnote{It is also possible to store vertices hashes.}
We assume that any round number during an execution can be expressed in
a constant number of bits, that is, the DAG never reaches round number
$2^{128}$ (note that round numbers grow slower than slot numbers).

We count the number of bits sent by correct processes in every round of
the DAG and divide it by the total number of ordered values therein.
The complexity of~\cite{cachin2005asynchronous} is $O(n^2\log(n) + nM)$, where $M$ is the
message (vertex) size.
Each message includes a set of proposed values and $n$ references,
and each reference includes a process id of size $\log(n)$.
Thus, if we batch $n\log(n)$ values in every message, the bit
complexity is $O(n^2\log(n) + 2n^2\log(n)) = O(n^2\log(n))$
for a broadcast.
%  and an amortized complexity of $O(n)$ per a value per broadcast.

Since each process is allowed to broadcast a single message in each
round, a correct process will not participate in more than $n$ reliable
broadcasts in a round, and thus the total bit complexity of correct
processes in a round is bounded by $O(n^3\log(n))$.
On the other hand, at least $2f+1 = O(n)$ vertices are ordered
% \sasha{Oded, is there a Lemma we can refer to here?} 
in every round.
Thus, $O(n^2\log(n))$ values are ordered in every round, which means
that the amortized communication complexity of \sys is $O(n)$.

\paragraph{Time complexity.}

By \Cref{claim:eventualCommit}, the number of waves, in expectation,
between two waves that satisfy the commit rule in $DAG_i$ for a correct
process $p_i$ is expected constant.
Since each wave consists of constant size chains of messages, by the
definition of time units, the number of time units, in expectation,
between two $p_i$'s commits is constant.
Every time $p_i$ commits a wave, it commits the wave's leader
causal history, which contains at least $O(n)$ proposals from
different correct processes.
Therefore, \sys's time complexity is $O(1)$ in expectation.

\section{Related Work}
\label{sec:related}

The first asynchronous Byzantine Agreement
protocols~\cite{ben1994power, rabin1983randomized} showed that the
FLP~\cite{fischer1985impossibility} impossibility result can be
circumvented with randomization.
Their communication and time complexity was exponential and a
significant amount of work has been done since then in attempt to
achieve optimal complexity under different assumptions. 
%  Later, a significant amount of work has been
% done to reduce their exponential .
Some works consider the information theoretical settings and present
protocols with polylogarithmic complexity that
tolerate adversaries with unbounded
computational power~\cite{kapron2010fast, patra2014asynchronous,
bangalore2018almost}.
Others follow a more practical approach and consider a computationally
bounded adversary in order to be able to use cryptographic primitives to improve
complexity~\cite{cachin2005Constantinople, cachin2001secure,
abraham2019vaba, lu2020dumbo}.
The pioneering crypto-based protocols~\cite{cachin2005Constantinople, cachin2001secure}
were later realized in HoneyBadgerBFT, the first asynchronous Byzantine SMR
system~\cite{miller2016honeybadgerBFT}.
However, while the state-of-the-art asynchronous Byzantine Agreement
protocols VABA~\cite{abraham2019vaba} and Dumbo~\cite{lu2020dumbo} rely
on cryptographic assumptions for both safety and liveness, \sys uses a hybrid
alternative by providing safety with information theoretical guarantees
and relying on cryptographic assumptions only for liveness. 

Many other works also presented protocols for the BAB problem in the asynchronous setting.
Some works like~\cite{reiter1994secure,kihlstrom2001securering} use cryptographic schemes for safety, and others like~\cite{correia2006consensus} do not use signatures.
Other works like~\cite{doudou2000abstractions} encapsulate timing
assumptions by relying on a failure detector.
All these works have higher expected communication complexity.

The idea of building a communication DAG and locally interpreting total
order was considered before~\cite{dolev1993early, chockler1998adaptive}.
To the best of our knowledge, the only algorithms that realize this idea
in the Byzantine settings are HashGraph~\cite{baird2016hashgraph} and
later Aleph~\cite{gkagol2019aleph}.
In contrast to \sys, HashGraph builds an unstructured DAG in which
processes (unreliably) send messages with 2 references to previous
vertices and on top of it run an inefficient binary agreement protocol,
which leads to expected exponential time complexity. 
Their communication complexity is not straightforward to analyze since they did
not clearly describe the mechanism that ensures that eventually all DAG
information is propagated to all processes, and no analysis is provided.
Aleph improves HashGraph's complexity by building a round-based DAG and
using a more efficient binary agreement
protocol~\cite{cachin2005Constantinople} to agree
on whether to commit every vertex in a round. 
% However, similarly to VABA and Dumbo based SMR, since they need all
% instances to terminate before they can totally order all vertices in a
% round, their time complexity is $O(\log(n))$.
They do not amortize complexity and have $O(n^3)$ per decision.
In contrast to \sys, both HashGraph and Aleph (1) do not
satisfy Validity; and (2) rely on signatures for safety and
thus are not post-quantum safe.

%\sasha{Need to cite more papers that solve Atomic Broadcast, e.g., ``From
%Consensus to Atomic Broadcast:Time-Free Byzantine-Resistant Protocols
%without Signatures'' and ``Secure and efficient asynchronous broadcast
%protocols``}
%\sasha{Check who cites them.}

% \paragraph{Coin} Cachin and a recent paper for adaptive coin by Julian.
% \sasha{Cite in Section 2}

% paragraph on different RB implementations.\sasha{not here}

% Cite MR for coin-based signature-free protocols. 

\section{Conclusion}
\label{sec:conclusion}
We presented \sys: an asynchronous Byzantine Atomic Broadcast protocol
with optimal resilience, optimal amortized communication complexity, and optimal 
time complexity.
\sys does not rely on cryptographic assumptions for safety. 
Instead, it rules out Byzantine equivocation by relying on the reliable
broadcast to guarantee that all correct processes eventually see the same
DAG.
Finally, we believe that \sys's elegant design, perfect load
balancing, and modular separation of concerns make it an adequate
candidate for future Byzantine SMR systems.

% ----------------------------------------------------------------
% you can include the acknowledgements in the source, but `anonymous' option will hide them
\begin{acks}
This work was funded by Novi, a Facebook subsidiary.
We also wish to thank the Novi Research team for valuable
feedback, and in particular George Danezis, Alberto Sonnino, and Dahlia
Malkhi.
\end{acks}
\clearpage
% ----------------------------------------------------------------
% use ACM-Reference-Format for the references
\bibliographystyle{ACM-Reference-Format}
\bibliography{references}

\end{document}